\documentclass[sigconf]{acmart}
\AtBeginDocument{%
  }

\copyrightyear{2027}
\acmYear{2027}
\setcopyright{cc}
\setcctype{by}
\acmConference[ASIA CCS '27]{22nd ACM Asia Conference on Computer and Communications Security}{July 12--16, 2027}{Macau, Macao}
\acmBooktitle{22nd ACM Asia Conference on Computer and Communications Security (ASIA CCS '27), July 12--16, 2027, Macau, Macao}
\acmDOI{10.1145/3862624.3872445}
\acmISBN{979-8-4007-3103-7/2027/07}

\usepackage{amsfonts}
\usepackage{amsthm}
\usepackage[ruled,noline,longend]{algorithm2e}
\DontPrintSemicolon
\SetAlCapHSkip{0pt}
\SetKwComment{Comment}{$\triangleright$\ }{}
\SetCommentSty{textrm}
\usepackage{bm}
\ifdefined\symbf\AtBeginDocument{\renewcommand{\bm}[1]{\symbfit{#1}}}\fi
\usepackage{textcomp}

\usepackage{pifont}
\usepackage{enumitem}
\usepackage{makecell}

\newcommand{\Z}{{\mathbb Z}}
\newcommand{\C}{{\mathbb C}}

\newcommand{\cC}{{\mathcal C}}
\newcommand{\cD}{{\mathcal D}}
\newcommand{\cE}{{\mathcal E}}
\newcommand{\cK}{{\mathcal K}}
\newcommand{\cQ}{{\mathcal Q}}
\newcommand{\cS}{{\mathcal S}}
\newcommand{\cV}{{\mathcal V}}

\newcommand{\Query}{\mathsf{Query}}
\newcommand{\Answer}{\mathsf{Answer}}

\usepackage{siunitx}
\usepackage{multirow}

\newcommand{\Enc}{\mathtt{Enc}}
\newcommand{\Dec}{\mathtt{Dec}}
\newcommand{\Gen}{\mathtt{KeyGen}}
\newcommand{\Eval}{\mathtt{Eval}}
\newcommand{\ctxt}{\mathtt{ctxt}}

\newcommand{\sk}{\mathtt{sk}}
\newcommand{\pk}{\mathtt{pk}}

\newcommand{\Match}{\mathsf{Match}}
\newcommand{\Mask}{\mathsf{Mask}}
\newcommand{\Comp}{\mathsf{Comp}}
\newcommand{\Decomp}{\mathsf{Decomp}}
\newcommand{\DB}{\mathsf{DB}}

\newcommand*{\comp}{\mathrm{c}}
\newcommand*{\compind}{\ensuremath{\mathrel{\overset{\comp}{\equiv}}}}

\newtheorem{remark}{Remark}

\begin{document}

\title{SIMD-Aware Homomorphic Compression and Application to Private Database Query}

\author{Jung Hee Cheon}
\orcid{0000-0002-7085-2220}
\affiliation{%
  \institution{Seoul National University}
  \city{Seoul}
  \country{Republic of Korea}
}
\affiliation{%
  \institution{CryptoLab Inc.}
  \city{Seoul}
  \country{Republic of Korea}
}
\email{jhcheon@snu.ac.kr}

\author{Keewoo Lee}
\authornote{Corresponding authors.}
\orcid{0000-0002-2521-8538}
\affiliation{%
  \institution{Ethereum Foundation}
  \city{Zug}
  \country{Switzerland}
}
\email{keewoo.lee@ethereum.org}

\author{Jai Hyun Park}
\orcid{0000-0002-5401-8949}
\affiliation{%
  \institution{CryptoLab Inc.}
  \city{Lyon}
  \country{France}
}
\email{jaihyunp@gmail.com}

\author{Yongdong Yeo}
\authornotemark[\value{footnote}]
\orcid{0009-0006-3826-4636}
\affiliation{%
 \institution{Seoul National University}
 \city{Seoul}
 \country{Republic of Korea}
}
\email{yongdong@snu.ac.kr}

\renewcommand{\shortauthors}{Cheon et al.}

\begin{abstract}
In a Private Database Query (PDQ) protocol, a server maintains a database, and users submit queries to retrieve records that satisfy specific conditions without revealing their interests to the server. Although homomorphic encryption (HE) provides a generic solution to PDQ, a naive approach leads to communication proportional to the database size. Consequently, a crucial step in HE-based PDQ protocols is homomorphic compression, which compresses encrypted sparse vectors representing query results with unknown non-zero indices into a dense form. However, prior compression schemes incur substantial computational overhead.

We propose a new homomorphic compression scheme based on Newton's identity. Our scheme is characterized by its simple compression procedure that can fully leverage the SIMD parallelism of modern HE schemes. When applied to PDQ, the compression reduces to homomorphic matrix-vector multiplications with only a single multiplicative depth, leading to significantly improved performance compared to prior approaches, as we experimentally demonstrate.
\end{abstract}


\maketitle

\section{Introduction}\label{sec:intro}

In a private database query scheme (PDQ) 
\cite{ACNS:BGHWW13, WYG+17_Splinter, CCS:AkaFelSha18, PoPETS:AGHL19, KLL+19, CCS:CDGLY21}, a server maintains a database, and users submit queries to retrieve records that satisfy specific conditions without revealing their interests to the server.\footnote{For relation of PDQ to other primitives (e.g., PIR), refer to \S\ref{ssec:related_primitives}.} 
For example, a user might want to fetch all financial/medical records with $\texttt{family\_name} = \texttt{Smith}$ but does not want to reveal the user's interest in those records.
To support various types of queries (e.g., exact-match, range, and wildcard queries, as well as their conjunctions and disjunctions) in a privacy-preserving manner, it is natural to consider employing homomorphic encryption (HE), a cryptosystem that enables computation over encrypted data. Indeed, a PDQ protocol can be constructed generically from HE, as described below and illustrated in Fig.~\ref{fig:PDQ_intro}.

\paragraph{HE-based PDQ and Homomorphic Compression}
Upon receiving the HE-encrypted query from the user (Steps~\ding{172} \& \ding{173} in Fig.~\ref{fig:PDQ_intro}), the server homomorphically evaluates the query on the database (Step~\ding{174}).
More specifically, the server first homomorphically checks the query-specified condition on all records in the database (Step~3-1. $\mathsf{Match}$).
Through $\mathsf{Match}$, the server obtains an encrypted binary vector, called an \emph{index vector}, where 1 indicates that the corresponding record meets the query condition.
The server then masks the database via homomorphic component-wise multiplication with the encrypted index vector (Step~3-2. $\mathsf{Mask}$). 
After $\mathsf{Match}$ and $\mathsf{Mask}$, the resulting ciphertext will be encrypting precisely the desired query result.

However, the size of the resulting ciphertexts produced by the $\mathsf{Mask}$ step scales linearly with the database size. 
That is, returning the result immediately after the $\mathsf{Mask}$ step incurs a higher communication cost than downloading the entire database from the server.
Consequently, one of the most crucial steps in HE-based PDQ protocols is the \emph{homomorphic compression}~\cite{CCS:CDGLY21, C:LiuTro22, EC:FleLarSim23}, in which the server compresses the ciphertext encrypting the masked database (Step~3-3. $\Comp$). 
Once a homomorphic compression scheme is available, the server can return the compressed ciphertext (Step~\ding{175}), and the user can decrypt and decompress it to extract the query result (Step~\ding{176}).
Ideally, the size of the compressed ciphertext should scale only with the number of records returned by the query and be independent of the database size.
Note that designing a homomorphic compression algorithm is nontrivial, as it must compress an encrypted vector without exact knowledge of the pertinent indices. 

\begin{figure*}
\begin{center}
\centerline{\includegraphics{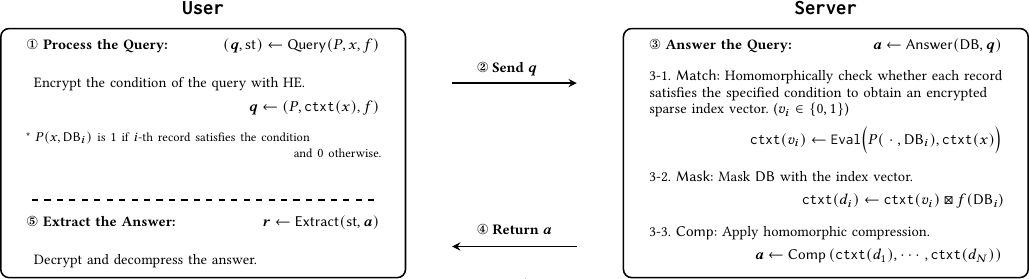}}
\caption{Overview of HE-based PDQ Design. Refer to \S\ref{sec:PDQ} for detailed and formal descriptions.}
\Description{A diagram of the PDQ protocol between a user and a server: the user sends an encrypted query, the server runs Match, Mask, and Comp on the database and returns a compressed answer, and the user decrypts and decompresses it.}
\label{fig:PDQ_intro}
\end{center}
\end{figure*}

\paragraph{Homomorphic SIMD}
Many modern HE schemes~\cite{ITCS:BraGenVai12, C:Brakerski12, EPRINT:FanVer12, AC:CKKS17} typically operate over large and complex plaintext spaces (e.g., $\Z_p[X]/\allowbreak(X^{8192}+1)$). 
In this regard, Smart--Vercauteren~\cite{SV14} introduced the idea of \emph{packing} several small messages (e.g., elements of $\Z_p$) into a plaintext.
Such packing enables \emph{homomorphic SIMD (Single Instruction Multiple Data)} optimization.
That is, we can homomorphically compute on \emph{multiple} messages simultaneously with a \emph{single} packed ciphertext. 
(See \S\ref{ssec:SIMD} for details on homomorphic SIMD.)
Homomorphic SIMD has become a standard technique, and the performance of HE applications is largely determined by how effectively it is used. Indeed, packing and SIMD optimization are common features of modern HE schemes (e.g., BFV~\cite{C:Brakerski12,EPRINT:FanVer12}, BGV~\cite{ITCS:BraGenVai12}, CKKS~\cite{AC:CKKS17}) and are supported by popular HE libraries (e.g., OpenFHE~\cite{WAHC:BBBCEG22}, HElib~\cite{HElib}, SEAL~\cite{SEAL}, HEaaN~\cite{HEaaN}).\footnote{One exception is TFHE~\cite{JC:CGGI20, TFHE}. Refer to \S\ref{ssec:related} for a discussion on the practicality of TFHE-based implementation.} However, previous works on homomorphic compression~\cite{CCS:CDGLY21, C:LiuTro22, EC:FleLarSim23, EC:GioGreSim26} largely neglect SIMD optimization. 
(See \S\ref{ssec:related}.)

\subsection{Our Contribution}\label{ssec:contribution}

We propose an efficient HE-based PDQ protocol (\S\ref{sec:PDQ}) tailored to modern FHE schemes such as the BFV scheme~\cite{C:Brakerski12,EPRINT:FanVer12}. Our construction is guided by the arithmetic circuit model, with particular attention to multiplicative depth and SIMD parallelism. At its core is a new homomorphic compression scheme that can fully leverage homomorphic SIMD optimization  (\S\ref{sec:homcomp}).

\paragraph{New Homomorphic Compression Scheme} 
We propose a new homomorphic compression scheme (\S\ref{sec:homcomp}), i.e., a homomorphic algorithm that compresses ciphertexts encrypting a length-$N$ $\Z_p$-vector with at most $s$ non-zero entries. Our scheme exhibits good asymptotic efficiency, independent of any SIMD optimization. The bounds below are \emph{non}-amortized.
  \begin{itemize}
      \item Our compression rate is asymptotically optimal: it achieves $N$-to-$2s$, which is within a factor of two of the lower bound.
      \item Our decompression complexity is $\mathsf{poly}(s, \log N)$.
  \end{itemize}
Furthermore, our scheme is fully deterministic, eliminating the complication and overhead of controlling failure probability.

More importantly, the practical speedup of our scheme stems from its friendliness toward modern FHE schemes:
  \begin{itemize}
      \item Our scheme admits an efficient implementation that fully exploits homomorphic SIMD.
      \item When applied to PDQ, our scheme consumes only a single multiplicative depth.
  \end{itemize}
  The first feature stems from the observation that the core of our compression algorithm reduces to matrix-vector multiplications. Implementing such matrix-vector products with homomorphic SIMD is itself a well-studied topic~\cite{C:HalSho14, USENIX:JuvVaiCha18, SP:LHHMQ21-fixed} (\S\ref{sec:opt_SIMD}). The second is achieved by folding the non-linear operations into the PDQ pipeline so that they impose no additional multiplicative depth. This is significant: the depth of the compression circuit directly drives parameter selection for the entire PDQ pipeline, so lower depth yields lighter parameters and better overall performance.

\paragraph{SIMD-Aware Implementation of PDQ}
Building upon our SIMD-aware homomorphic compression scheme, we propose an efficient HE-based PDQ protocol (\S\ref{sec:PDQ}). 
We provide implementation details with particular attention to depth optimization (\S\ref{sssec:opt_reuse}) and homomorphic SIMD (\S\ref{sec:opt_SIMD}).
This includes, for instance, our decision to follow the approach of \textsc{Pegasus}~\cite{SP:LHHMQ21-fixed} for homomorphic matrix-vector multiplication.

We further adopt the \emph{ring-switching} technique~\cite{SCN:GHPS12, GHPS13} as a practical optimization (\S\ref{subsec:RingSwitching}): after the $\sf Mask$ step (Fig.~\ref{fig:PDQ_intro}), we switch to smaller HE parameters, improving both communication and computation efficiency. In the PDQ setting, large HE parameters are required to support practically relevant queries, but this coarsens ciphertext granularity, inflating both the digest size and the runtime of the homomorphic compression step. Prior work, particularly \cite{CCS:CDGLY21}, overlooks this tension.

\paragraph{Concrete Efficiency}
We experimentally demonstrate the concrete efficiency of our PDQ protocol for the exact-match query case (\S\ref{sec:experiment}). For instance, on a 16-bit key-value database with $N=2^{14}$ records and $s=32$ matching entries, our PDQ protocol generates the entire PDQ answer, a single $129$\,KB ciphertext, in $2.56$ seconds end-to-end on a single thread. Of this, $2.19$ seconds is spent in $\Match$, while the remaining $0.37$ seconds covers $\Mask$, ring-switching, and $\Comp$ combined.

To put these numbers in context, our compression achieves over $100\times$ runtime improvement compared to the Bloom-filter-based approach of Choi et al.~\cite{CCS:CDGLY21}. Compared to the TFHE-based homomorphic sorting approach of Parbat and Chatterjee~\cite{ToSC:ParCha25}, our scheme achieves at least a $500{,}000\times$ end-to-end speedup and a $7.5\times$ reduction in communication.

\subsection{Technical Overview}

The high-level idea of our compression scheme is to apply the \emph{power-sum method}, which encodes a vector $(v_1,\cdots, v_N)$ into the values $\sum_{i=1}^N i^j \cdot v_i$, to compress encrypted sparse vectors (or the homomorphically masked database in the context of PDQ).
Even though a previous work~\cite{CCS:CDGLY21} used the power-sum method to compress sparse \emph{index} vectors (whose components are either 0 or 1), to the best of our knowledge, our work is the first to use the method for general sparse vectors.

To describe our scheme, for the moment, suppose that the corresponding index vector is provided together with the sparse vector. 
The $i$-th component of the corresponding index vector is $1$ if and only if the $i$-th component of the sparse vector is non-zero. 
Our main idea is to leverage this index vector as a hint while decoding the sparse vector. 

The key observation is that the power-sum encoding can be interpreted as a matrix multiplication by a Vandermonde-like matrix. 
This observation also allows us to fully leverage homomorphic SIMD in our scheme.

More precisely, to compress a sparse vector, we encode both the sparse vector and its index vector using the power-sum method. 
To decompress, we first decode the index vector via Newton's identity (\S\ref{ssec:newton}) as done in \cite{CCS:CDGLY21}.
Then, we can use this index vector to decode the sparse vector. 
To decode the sparse vector, we consider the submatrix of the Vandermonde-like matrix corresponding to the decoded index vector. 
When $p > N$, where $\Z_p$ is the message space,\footnote{This constraint can be relaxed via $p$-ary decomposition of indices $1\le i\le N$, which incurs only $\log_p N$ overhead. See Rem.~\ref{rmk:our_p_geq_N}.} this submatrix is invertible, yielding a linear system that we solve efficiently via the Bj\"orck--Pereyra method. We refer to \S\ref{sec:homcomp} for a more detailed and formal description.

It remains to answer how to obtain the corresponding index vector of the given sparse vector. 
In general, we can homomorphically compute the index vector from the sparse vector.
For example, with the finite field message space $\Z_p$, we can leverage Fermat's little theorem (Alg.~\ref{alg:Comp}).

However, when applying our compression scheme in the context of PDQ (i.e., in $\Comp$), we do not have to bother with this issue. 
We already have the corresponding index vector for the sparse vector (or masked database) from $\mathsf{Match}$. 
The punchline is that, in natural HE-based PDQ protocols described in Fig.~\ref{fig:PDQ_intro}, the sparse vector is obtained \emph{from} the index vector through masking ($\mathsf{Mask}$). 
We believe that this is also true in other scenarios where homomorphic compression is needed and that our compression scheme will demonstrate good performance in general.

\subsection{Problem Settings}\label{ssec:example_application}

\paragraph{Database: Encrypted vs. Plaintext}
PDQ can be applied to both settings in which the server stores the database in either encrypted or plaintext form. One important scenario in the former setting is outsourced database services, where clients upload encrypted data to an untrusted server and wish to hide both the database contents and their access patterns. In this setting, an HE-based PDQ scheme enables the server to process encrypted queries directly over HE-encrypted databases, allowing clients to retrieve only the relevant records while preserving privacy. PDQ can also be applied when the database is maintained in plaintext by the server, for example when the server owns the database or when the data is publicly available. 

In this work, we present the main ideas in the plaintext database setting for readability and clarity. However, in our experiments (\S\ref{sec:experiment}), we benchmark the encrypted database setting, which is at least as expensive as the plaintext setting.

\paragraph{Privacy: One-Sided vs. Two-Sided}
We say a PDQ protocol provides \emph{one-sided} privacy when it hides the client's query and access pattern from the server, and \emph{two-sided} privacy when, in addition, nothing about the server's database is revealed to the client beyond the query result itself. (See Def.~\ref{def:2-sided}.)

Two-sided privacy becomes relevant primarily when the server owns the database and holds it in plaintext, but must keep its contents confidential from the client due to, e.g., intellectual-property or regulatory constraints. One example is an investigative agency querying transaction records pertaining to specific suspects: the database owner must not disclose records of unrelated individuals, while the agency still receives the matching ones. Such settings can be supported on top of our HE-based construction by applying standard \emph{circuit-privacy} techniques (Rem.~\ref{rem:2-sided-privacy}).

On the other hand, one-sided privacy is sufficient when the database is already public, or known to the client (as in the outsourced-database setting, where the client uploaded the database itself). One concrete example of the former is Oblivious Message Retrieval (OMR)~\cite{C:LiuTro22}, which we discuss in more detail in \S\ref{sssec:OMR}. In this work, we focus on the one-sided privacy setting.

\subsection{Related Work}\label{ssec:related}

\subsubsection{Private Database Query}
PDQ and closely related scenarios have been studied extensively~\cite{ACNS:BGHWW13, WYG+17_Splinter, CCS:AkaFelSha18, PoPETS:AGHL19, KLL+19, CCS:CDGLY21, SAC:CGKP24, ToSC:ParCha25}. However, most prior work does not consider homomorphic compression, and as a result either (1) supports only limited functionality or (2) incurs high communication or round complexity. For instance, the schemes of~\cite{CCS:AkaFelSha18, PoPETS:AGHL19} support only queries that return a single data point, whereas our scheme handles general queries with possibly multiple matches. Splinter~\cite{WYG+17_Splinter} requires multiple non-colluding servers, each holding a copy of the database, whereas our scheme works against a single adversarial server.

Schemes in~\cite{SAC:CGKP24, ToSC:ParCha25} indirectly address homomorphic compression in PDQ by applying homomorphic sorting to remove non-matching entries. Parbat and Chatterjee~\cite{ToSC:ParCha25} used a bubble-sort-based approach with $O(N^2)$ comparisons, while Cong et al.~\cite{SAC:CGKP24} used a top-$s$ selection algorithm reducing this to $O(N \log^2 s)$. However, homomorphic comparison is a heavy hammer in arithmetic-circuit HE schemes such as BFV~\cite{C:Brakerski12,EPRINT:FanVer12}, and even when instantiated over TFHE~\cite{JC:CGGI20}, these approaches remain significantly less efficient than ours. See \S\ref{sssec:comp_SHELDB} for a detailed comparison.

\subsubsection{Homomorphic Compression}\label{sssec:related_comp}

There have been several works focusing on the design of homomorphic compression schemes. 
See \S\ref{subsec:comparison} and Tab.~\ref{tab:comparison} for an asymptotic comparison.

\begin{itemize}[leftmargin=*, itemsep=0.5em]
    \item Choi et al.~\cite{CCS:CDGLY21} proposed a probabilistic scheme based on a Bloom-filter-like data structure, namely BFS-CODE. 
    They also proposed two schemes, BF-COIE and PS-COIE, that can only compress ciphertexts encrypting \emph{index} vectors, whose entries are either $0$ or $1$. 
    See \S\ref{sssec:comp_CDGLY} for a detailed empirical comparison.

    \item Fleischhacker--Larsen--Simkin~\cite{EC:FleLarSim23} proposed two schemes.
    The first one is based on fast Fourier transform (FFT). 
    It enjoys an asymptotically optimal compression rate, but its decompression complexity is as large as $O(s\sqrt{N})$, where $N$ is the size of the database.
    The second one is based on a Bloom-filter-like data structure with false positives.

    \item Liu--Tromer~\cite{C:LiuTro22} constructs a probabilistic homomorphic compression scheme in the process of designing an \emph{oblivious message retrieval} scheme.
    See \S\ref{sssec:OMR} for a discussion on empirical comparison.

    \item Giorgi--Grenet--Simkin~\cite{EC:GioGreSim26} is a concurrent work with ours and focuses on a theoretical construction of a homomorphic compression scheme. Asymptotically, their scheme achieves favorable computation complexities. However, unlike our work, it does not take into account practical considerations of modern HE schemes, such as multiplicative depth and SIMD parallelism. 
    They do not provide an implementation or empirical evaluation, making it difficult to measure its practical performance.

    \item Bienstock--Patel--Seo--Yeo~\cite{USENIX:BPSY24} pursue a similar goal of constructing a homomorphic compression scheme, but under a different setting in which the party performing decompression already knows the indices of interest, as in batch PIR. In general, this assumption does not hold in the PDQ setting, where the client does not know the relevant indices in advance.

    \item Akavia--Shaul--Shayevitz~\cite{WAHC:AkaShaSha26} consider a different problem in which decompression must also be performed under homomorphic encryption.
    This is unnecessary in the PDQ setting, where the client decrypts and decompresses in the clear, and it comes at a considerably higher cost.
    Furthermore, their schemes are restricted to sparse \emph{index} vectors.
\end{itemize}

\subsection{Related Primitives}\label{ssec:related_primitives}

\subsubsection{Searchable Encryption} Searchable Encryption~\cite{SP:SonWagPer00,SP:FVYSHG17} supports efficient search over encrypted data.
However, it requires extensive pre-processing and compromises privacy for efficiency, allowing hard-to-analyze leakage.
In contrast, our HE-based approach leaks nothing more than a predefined upper bound on the number of retrieved data points. 

\subsubsection{Private Information Retrieval (PIR)} 
PIR~\cite{FOCS:CGKS95} allows a client to retrieve a data item from a server without revealing which item is accessed. However, PIR is limited to retrieving an array entry by index. In contrast, PDQ (and our HE-based approach in particular) supports arbitrarily complex queries over relational tables with multiple attributes, and the response is not restricted to a single entry.

\subsubsection{Compressible HE}
Compressible HE~\cite{TCC:BDGM19,TCC:GenHal19} focuses on transforming ciphertexts into \emph{high-rate} representations, i.e., maximizing the plaintext-to-ciphertext size ratio. In contrast, our notion of homomorphic compression aims to compress the underlying message itself by exploiting redundancy in the plaintext, while operating under homomorphic encryption.

\section{Preliminaries}\label{sec:prelim}

\subsection{Notations}\label{ssec:notation}

We use $[n]$ to denote the set $\{1, 2, \cdots , n\}$.
For a message $\mu$, we use $\ctxt(\mu)$ to denote a ciphertext encrypting the message $\mu$.
We use lowercase letters (e.g., $v$) to denote scalars, boldface lowercase letters (e.g., $\bm{v}$) to denote vectors, and boldface uppercase letters (e.g., $\bm{M}$) to denote matrices.
For a matrix $\bm{M}$, we will sometimes abuse the notation $\bm{M}$ to denote the matrix itself or a circuit for computing the multiplication by $\bm{M}$. 

\subsection{Homomorphic Encryption}
A homomorphic encryption (HE) scheme allows us to compute an encryption of $f(m)$ when only given $f$ and a ciphertext encrypting $m$ without decryption.
An HE scheme $\mathcal E$, with a message space $\mathcal{M}$ and a ciphertext space $\mathcal{C}$, is a tuple of probabilistic polynomial time algorithms $(\Gen,\, \allowbreak \Enc,\, \Dec,\, \Eval)$ described below:

\begin{itemize}
    \item $\Gen(1^\lambda)$: The key generation algorithm takes a security parameter $1^\lambda$ as an input, and returns a secret key $\sk$ and public key $\pk$. 
    \item $\Enc_{\pk}(m)$: The encryption algorithm takes the public key $\pk$ and a message $m\in \mathcal{M}$ as inputs, and returns a ciphertext $\ctxt(m) \in \mathcal{C}$ encrypting $m$.
    \item $\Dec_{\sk}(\ctxt(m))$: The decryption algorithm takes the secret key $\sk$ and a ciphertext $\ctxt(m) \in \mathcal{C}$, and returns a message $m\in \mathcal{M}$.
    \item $\Eval_{\pk}(f, \ctxt(m))$: The evaluation algorithm takes the public key $\pk$, a circuit $f$, and a ciphertext $\ctxt(m)\in \mathcal{C}$ as inputs, and returns a ciphertext $\ctxt(f(m))\in \mathcal{C}$ that encrypts $f(m)\in \mathcal{M}$ as a message.
\end{itemize}
Throughout the paper, for a message vector $\bm{m}=(m_1,\cdots,m_{\ell})\in \mathcal{M}^{\ell}$, we denote $\ctxt(\bm{m})$ as a vector of ciphertexts $(\ctxt(m_1),\cdots \allowbreak, \ctxt(m_{\ell}))\in \mathcal{C}^{\ell}$.
We use $\boxtimes$ to denote homomorphic multiplication. 

\subsection{The BFV Scheme and Its SIMD Structure}\label{ssec:SIMD}

In this work, we instantiate FHE using the BFV scheme~\cite{C:Brakerski12,EPRINT:FanVer12}. The BFV scheme supports homomorphic SIMD arithmetic in $\Z_p$ for a prime $p$, which is required in our homomorphic compression scheme.
In particular, we consider the BFV scheme with the ring degree $n$ that is a power of $2$ and the plaintext modulus $p$ satisfying $p=1\pmod{2n}$.
Then, the plaintext space is $\mathcal{R}_p := \Z_p [X] / (X^{n}+1)\cong \Z_p^{2\times(n/2)}$, and the ciphertext space is $\mathcal{R}_q^2$, where $\mathcal{R}_q:=\Z_q [X] / (X^{n}+1)$ for a ciphertext modulus $q$.
Let $\bm{m}\in \Z_p^{2\times(n/2)}$ be a plaintext and denote its rows and entries as follows.
We often call each entry of a plaintext a \emph{slot}.
\[
    \bm{m}= 
    \begin{bmatrix}
        \bm{m}_1^T\\
        \bm{m}_2^T
    \end{bmatrix}
    =
    \begin{bmatrix}
        m_{1,1} & \cdots & m_{1, n/2}\\
        m_{2,1} & \cdots & m_{2, n/2}
    \end{bmatrix}.
\]

Homomorphic SIMD is provided by (1) BFV SIMD arithmetic in $\Z_p$ and (2) BFV homomorphic rotations (horizontal and vertical) to support its utilization. 
That is, we can homomorphically add and multiply a ciphertext encrypting $\bm{v}\in \Z_p^{2\times(n/2)}$ with a ciphertext (or plaintext) of $\bm{w}\in \Z_p^{2\times(n/2)}$ to obtain a ciphertext of $\bm{v}+\bm{w}$ and $\bm{v}\odot\bm{w}$, respectively.
Here, $\odot$ denotes Hadamard multiplication between two vectors. 
For homomorphic rotations, to be precise, let $\mathtt{Rot}( \;\cdot\; ; r)$ denote the rotation operation which maps $[v_1, \cdots , v_{n/2}]\allowbreak \mapsto [v_{1+r},\cdots, v_{n/2}, v_{1}, \cdots , v_{r}]$.
Then, the BFV scheme supports homomorphic rotations computing $\mathtt{Rot}_{\text{row}}$ and $\mathtt{Rot}_{\text{col}}$, defined as follows.

\begin{align*}
    \mathtt{Rot}_{\text{row}}\left( 
    \begin{bmatrix}
        \bm{m}_1^T\\
        \bm{m}_2^T
    \end{bmatrix}
    ; r \right)
    &=\begin{bmatrix}
        \mathtt{Rot}(\bm{m}_1^T; r)\\
        \mathtt{Rot}(\bm{m}_2^T; r)
    \end{bmatrix},\\
    \mathtt{Rot}_{\text{col}}\left( 
    \begin{bmatrix}
        \bm{m}_1^T\\
        \bm{m}_2^T
    \end{bmatrix}
    ; 1 \right)
    &=\begin{bmatrix}
        \bm{m}_2^T\\
        \bm{m}_1^T
    \end{bmatrix}.
\end{align*}

\subsection{Newton's Identity}\label{ssec:newton}
We review \emph{Newton's identity}, which bridges \emph{power sum} and \emph{elementary symmetric polynomials}.
Newton's identity is the main tool in our homomorphic compression scheme.

\begin{definition}[Power Sum Polynomial]\label{def:powersumpoly}
    Let $R$ be a ring. 
    For $n\ge 1$ and $k\ge 0$, we denote and define the $k$-th \emph{power sum polynomial} in $R[X_1, \cdots, X_n]$ as
    \[
        p_k = \sum_{i=1}^{n}X_i^k.
    \]
\end{definition}

\begin{definition}[Elementary Symmetric Polynomial]\label{def:elementarysympoly}
    Let $R$ be a ring. 
    For $n\ge 1$ and $0 \le k \le n$, we define the $k$-th \emph{elementary symmetric polynomial} $e_k$ in $R[X_1, \cdots, X_n]$ as the sum of all distinct products of $k$ distinct variables, i.e.,
    \begin{align*}
        e_0 &= 1,\\
        e_1 &= \sum_{1 \le i \le n} X_i,\\
        e_2 &= \sum_{1 \le i < j \le n} X_i X_j,\\
        &\quad \vdots\\
        e_n &= X_1 X_2 \cdots X_n.
    \end{align*}
\end{definition}

\begin{lemma}[Newton's Identity]\label{lem:newton}
    Let $R$ be a ring. 
    In the ring $R[X_1, \cdots, X_n]$, the following identity holds for all $1 \le k \le n$.
    \begin{equation*}
        k \cdot e_k = \sum_{i=1}^{k} (-1)^{i-1} e_{k-i} \cdot p_{i}
    \end{equation*}
\end{lemma}

\section{Homomorphic Compression}\label{sec:homcomp}

In this section, we present our homomorphic compression scheme.
We remark that the SIMD structure is \emph{not} considered throughout this section. 
We begin by formally defining the notion of \emph{homomorphic compression}.
Let $\mathcal{E}$ be an HE scheme with a message space $\mathcal{M}$ and a ciphertext space $\mathcal{C}$.

\begin{definition}[Homomorphic Compression]
    Let $\Omega$ be a subset of $\mathcal{M}^N$.
    An $N$-to-$s$ \emph{homomorphic compression scheme} for $\Omega$ over $\mathcal{E}$ is a pair of algorithms $(\Comp_{\pk}, \mathsf{Decomp}_{\sk})$, for a pair of keys $(\sk, \pk)$ of $\mathcal{E}$. 
    The compression algorithm $\Comp_{\pk}:\mathcal{C}^N \rightarrow \mathcal{C}^s$
    only takes the public key $\pk$ of $\mathcal{E}$ and ciphertexts $\ctxt(\bm{a})$ of some $\bm{a}\in \Omega$, and the decompression algorithm $\mathsf{Decomp}_{\sk}:\mathcal{C}^s \rightarrow \mathcal{M}^N$ takes the secret key $\sk$ of $\mathcal{E}$ and a compressed ciphertext, such that    
\[
    \mathsf{Decomp}_{\sk}\left(\Comp_{\pk}\left(\ctxt(\bm{a})\right)\right)=\bm{a}, \quad \forall \bm{a}\in \Omega.
\]
\end{definition}

Throughout the paper, we often omit subscripts corresponding to the public and secret keys ($\pk$ and $\sk$) when they are clear from context.

In this work, we are especially interested in homomorphic compression for \emph{sparse} vectors, which are defined below.

\begin{definition}[Index Set]\label{def:indexset}
    Let $R$ be a ring and $\bm{v}=(v_1,\cdots,v_N) \in R^N$ be a vector.
    We define $I_{\bm{v}}\subset [N]$, the \emph{index set} of $\bm{v}$, as follows. 
    \[
        i\in I_{\bm{v}} \iff v_i \neq 0
    \]
\end{definition}

\begin{definition}[$s$-Sparsity]
    Let $R$ be a ring and $\bm{v} \in R^N$ be a vector.
    We say that $\bm{v}$ is \emph{$s$-sparse} if $|I_{\bm{v}}|\le s$ holds, i.e., the number of non-zero components in $\bm{v}$ is at most $s$.
\end{definition}

\subsection{Compression for Index Vectors}
\label{ssec:compress-index}

First, we present a homomorphic compression scheme for \emph{index} vectors, defined below.  

\begin{definition}[Index Vector]\label{def:indexvec}
    Let $R$ be a ring.
    We say that $\bm{v}\in R^N$ is an \emph{index vector} if every component of $\bm{v}$ is either $0$ or $1$. 
    Furthermore, we say that $\bm{v}\in R^N$ is the index vector of $\bm{w} \in R^N$ if $\bm{v}$ is the index vector such that $I_{\bm{v}} = I_{\bm{w}}$.
\end{definition}

The main tool in this section is Newton's identity (\S\ref{ssec:newton}) and the following Vandermonde-like matrix $\bm{C} \in \Z_p^{s\times N}$, where $p$ is a prime.
Note that $s$ arbitrary columns of $\bm{C}$ are linearly independent if $p > N$.

\begin{align}
    \bm{C} :=
    \begin{bmatrix}
        1 & 2 & \cdots & N\\
        1^2 & 2^2 & \cdots & N^2\\
        \vdots & \vdots & & \vdots\\
        1^{s} & 2^{s} & \cdots & N^{s}
    \end{bmatrix} \pmod{p}\label{eq:vandermonde}
\end{align}

\begin{algorithm}[t]
\caption{$\mathsf{ReconstIdx}$}\label{alg:ReconstIdx}
\textbf{Input:} $\bm{w}=(w_1,\cdots,w_s) \in \Z_p^s$\;
\smallskip
\textbf{Step 1}: Compute $f_{\bm{v}}^s(X)$ from $\bm{w}$ using Newton's identity.\;
\Indp
$a_0 \gets 1$\;
\For{$k = 1, \cdots, s$}{
$a_k \gets k^{-1} \cdot \sum_{i=1}^{k} (-1)^{i-1} a_{k-i} \cdot w_{i}$\Comment*[r]{see Lem.~\ref{lem:newton}}
}
$f_{\bm{v}}^s(X) \gets \sum_{k=0}^s (-1)^k a_k X^{s-k}$\;
\Indm
\smallskip
\textbf{Step 2}: Reconstruct $I_{\bm{v}}$ from $f_{\bm{v}}^s(X)$.\;
\Indp
$g_{\bm{v}}^s(X) \gets f_{\bm{v}}^s(X)$\;
\While{$X$ divides $g_{\bm{v}}^s(X)$}{
$g_{\bm{v}}^s(X) \gets g_{\bm{v}}^s(X)/X$\;
}
$I_{\bm{v}} \gets$ Find zeros of $g_{\bm{v}}^s(X)$\Comment*[r]{e.g., Cantor--Zassenhaus}
\Indm
\textbf{return} $I_{\bm{v}} \subset [N]$\Comment*[r]{set of nonzero indices of $\bm{v}$ (Def.~\ref{def:indexset})}
\end{algorithm} 

\begin{algorithm}[t]
\caption{$\mathsf{CompIdx}$}\label{alg:CompIdx}
\smallskip
\textbf{Input:} $\ctxt(\bm{v})$\Comment*[r]{ciphertexts of an index vector $\bm{v}\in \Z_p^N$}
\smallskip
$\ctxt(\bm{w}) \gets \Eval(\bm{C}, \ctxt(\bm{v}))$\Comment*[r]{see Eq.~\eqref{eq:vandermonde} and \S\ref{ssec:notation}}
\smallskip
\textbf{return} $\ctxt(\bm{w})$\Comment*[r]{ciphertexts of compressed index $\bm{w}\in\Z_p^s$}
\end{algorithm}

\begin{algorithm}[t]
\caption{$\mathsf{DecompIdx}$}\label{alg:DecompIdx}
\smallskip
\textbf{Input:} $\ctxt(\bm{w})$\Comment*[r]{ciphertexts of $\bm{w}\in\Z_p^s$}
\smallskip
\textbf{Step 1}: Reconstruct $I_{\bm{v}}$ from $\ctxt(\bm{w})$.\;
\Indp
$\bm{w} \gets \Dec(\ctxt(\bm{w}))$\;
$I_{\bm{v}} \gets \mathsf{ReconstIdx}(\bm{w})$\Comment*[r]{see Alg.~\ref{alg:ReconstIdx}}
\Indm
\smallskip
\textbf{Step 2}: Reconstruct $\bm{v}$ from $I_{\bm{v}}$.\Comment*[r]{redundant step for formatting}
\smallskip
\Indp
$\bm{v} \gets$ $(v_1,\cdots,v_N)$, where $v_i =\begin{cases} 1, & \text{if $i \in I_{\bm{v}}$}\\ 0, & \text{otherwise.} \end{cases}$\;
\Indm
\textbf{return} $\bm{v}\in \Z_p^N$\Comment*[r]{recovered index vector}
\end{algorithm}

\begin{lemma}\label{lem:ReconstIdx}
    Let $\bm{v} \in \Z_p^N$ be an $s$-sparse index vector, where $p>N$ is a prime.
    Furthermore, let $\bm{w}=\bm{C}\bm{v}$ (Eq.~\eqref{eq:vandermonde}).
    Then, we can reconstruct $I_{\bm{v}}\subset [N]$, the index set of $\bm{v}$, given only $\bm{w}\in\Z_p^s$. 
    In particular, using $\mathsf{ReconstIdx}$ (Alg.~\ref{alg:ReconstIdx}), this costs $O(s^2 + s \log^2 s \cdot (\log s + \log p))$ operations in $\Z_p$.
\end{lemma}

\begin{proof}
    The lemma is direct from Alg.~\ref{alg:ReconstIdx}. 
    Here, we give an overview of how and why Alg.~\ref{alg:ReconstIdx} works.
    Observe that, for $\bm{w}=(w_1,\cdots,w_s)$ and $\bm{v}=(v_1,\cdots,v_N)$,
    \[
        w_j = \sum_{i=1}^N i^j\cdot v_i=\sum_{i\in I_{\bm{v}}}i^j \pmod{p}.
    \]
    That is, \(w_j\) represents the evaluated value of the \(j\)-th power sum polynomial in $\mathbb{Z}_p[X_1, \cdots, X_s]$
    (Def.~\ref{def:powersumpoly}) evaluated at \emph{multiset} $J_{\bm{v}}^s = I_{\bm{v}} \cup \{0,\cdots,0\}$, where $|J_{\bm{v}}^s| = s$ and the multiplicity of $0$ is $s-|I_{\bm{v}}|$.

    Now consider the following degree-$s$ polynomial whose product is over \emph{multiset} $J_{\bm{v}}^s$.
    \[
        f_{\bm{v}}^s(X):= \prod_{i \in J_{\bm{v}}^s} (X-i)
    \]
    Note that, through Vieta's formulas, we can relate the coefficients of $f_{\bm{v}}^s(X)$ with elementary symmetric polynomials in $\Z_p[X_1, \cdots, X_s]$ (Def.~\ref{def:elementarysympoly}) evaluated at  $J_{\bm{v}}^s$. 
    Thus, by recursively applying Newton's identity (\S\ref{ssec:newton}), we can compute the coefficients of $f_{\bm{v}}^s(X)$ from $w_j$'s. 
    For a precise description of the procedure, refer to Step~1 of Alg.~\ref{alg:ReconstIdx}.
    This procedure costs $O(s^2)$ operations in $\Z_p$.

    Then, we can reconstruct $I_{\bm{v}}$ from $f_{\bm{v}}^s(X)$ through root-finding.
    For a precise description of the procedure, refer to Step~2 of Alg.~\ref{alg:ReconstIdx}.
    This procedure costs $O(s \log^2 s \cdot (\log s + \log p))$ operations in $\Z_p$ with the Cantor--Zassenhaus algorithm~\cite{GG13_MCA}.
\end{proof}

Let $\mathcal{E} = (\Gen,\, \allowbreak \Enc,\, \Dec,\, \Eval)$ be an HE scheme with message space $\Z_p$. 

\begin{theorem}[Homomorphic Compression for Index Vectors]\label{thm:CompIdx}
The pair of algorithms $(\mathsf{CompIdx}, \allowbreak \mathsf{DecompIdx})$ (Algs.~\ref{alg:CompIdx} and~\ref{alg:DecompIdx}) is an $N$-to-$s$ homomorphic compression scheme for $s$-sparse index vectors in $\Z_p^N$ (over $\mathcal{E}$).
Furthermore,
\begin{enumerate}
    \item $\mathsf{CompIdx}$ (Alg.~\ref{alg:CompIdx}) costs $O(N \cdot s)$ homomorphic additions and multiplications.
    \item Step~1 of $\mathsf{DecompIdx}$ (Alg.~\ref{alg:DecompIdx}) costs $s$ decryptions and $O(s^2 + s \log^2 s \cdot (\log s + \log p))$ operations in $\Z_p$.
\end{enumerate}
\end{theorem}
\begin{proof}
The theorem is direct from Lem.~\ref{lem:ReconstIdx}. 
\end{proof}

We note that our scheme for index vectors is essentially the same as the PS-COIE of \cite{CCS:CDGLY21}.
Our novelty lies in interpreting their \emph{power-sum} encoding as a matrix multiplication by the Vandermonde-like matrix $\bm{C}$. 
This observation later allows us to extend the scheme into a compression scheme that also works for non-index vectors (\S\ref{ssec:compress-sparse}) and to fully leverage homomorphic SIMD structure (\S\ref{sec:opt_SIMD}).

\subsection{Compression for Sparse Vectors}
\label{ssec:compress-sparse}

Now, we construct a new homomorphic compression scheme for general sparse vectors with entries that are not necessarily $0$ or $1$.
Our main idea is to compress a general sparse vector $\bm{d}$ as $\bm{C}\bm{d}$ as before, while using its index vector (Def.~\ref{def:indexvec}) as a \emph{hint} when decompressing. 
We begin with a lemma.

\begin{algorithm}[t]
\caption{$\mathsf{Reconst}$}\label{alg:Reconst}
\smallskip
\textbf{Input:} $\bm{e}\in \Z_p^s$, $I_{\bm{d}}=\{i_1, \cdots, i_\ell\}\subset [N]$\Comment*[r]{$|I_{\bm{d}}| = \ell \le s$}
\medskip
$\hat{\bm{C}} \gets [\bm{C}_{i_1} | \cdots | \bm{C}_{i_\ell}]$\Comment*[r]{$k$-th column of $\hat{\bm{C}}$ is $i_k$-th column of $\bm{C}$}
\smallskip
$\hat{\bm{d}}=(\hat{d}_1, \cdots, \hat{d}_\ell) \gets$ Solve $\hat{\bm{C}} \bm{x} = \bm{e}$ for $\bm{x}$ using Bj\"orck--Pereyra\;
\smallskip
$\mathsf{sparse}(\bm{d}) \gets \{ (i_1,\hat{d}_1), \cdots, (i_\ell,\hat{d}_\ell) \}$\;
\medskip
\textbf{return} $\mathsf{sparse}(\bm{d})$\Comment*[r]{recovered entries (Lem.~\ref{lem:Decomp})}
\end{algorithm}

\begin{lemma}\label{lem:Decomp}
    Let $\bm{d}=(d_1,\cdots, d_N) \in \Z_p^N$ be an $s$-sparse vector, where $p > N$ is a prime. 
    Furthermore, let $\bm{e} = \bm{C} \bm{d}$ (Eq.~\eqref{eq:vandermonde}).
    We can reconstruct the following \emph{sparse} representation of $\bm{d}$, given only $\bm{e}\in \Z_p^s$ and $I_{\bm{d}}\subset [N]$, the index set of $\bm{d}$ (Def.~\ref{def:indexset}).
    \[
        \mathsf{sparse}(\bm{d})
        := \left\{ (i,d_i): i\in I_{\bm{d}}\right\}
    \]
    In particular, using $\mathsf{Reconst}$ (Alg.~\ref{alg:Reconst}), this costs $O(s^2)$ operations in $\Z_p$.
\end{lemma}

\begin{proof}
The lemma is direct from Alg.~\ref{alg:Reconst}. 
Let  $I_{\bm{d}} = \{i_1, \cdots , i_{\ell}\}$, where $\ell = |I_{\bm{d}}|$, and let $\bm{C}_k$ denote the $k$-th column of $\bm{C}$, i.e., $\bm{C}=[\bm{C}_1|\cdots |\bm{C}_N]$.
Then, set $\hat{\bm{C}} = [\bm{C}_{i_1} | \cdots | \bm{C}_{i_\ell}]$.
Note that, since $\ell \le s$, the linear equation $\hat{\bm{C}}\bm{x}=\bm{e}$ has at most one solution. 
It is easy to check that $\hat{\bm{d}}=(d_{i_1}, \cdots , d_{i_{\ell}})$ is the unique solution of the equation as follows. 
\[
    \hat{\bm{C}}\hat{\bm{d}} 
    = \sum_{k=1}^\ell d_{i_k} \cdot \bm{C}_{i_k}
    = \sum_{i=1}^N d_i \cdot \bm{C}_i
    = \bm{C} \bm{d}
    = \bm{e}
\]
The runtime of Alg.~\ref{alg:Reconst} is dominated by finding the solution of $\hat{\bm{C}}\bm{x}=\bm{e}$, using the Bj\"orck--Pereyra method, which requires $O(s^2)$ operations in $\mathbb{Z}_p$.
\end{proof}

Let $\mathsf{Power}_{\ell}$ denote an arithmetic circuit over $\Z_p$ such that 
\[
    \mathsf{Power}_{\ell} (\bm{v}) = (v_1^{\ell},\cdots,v_N^{\ell}),
\]
where $\bm{v} = (v_1, \cdots , v_N) \in \Z_p^N$. 
Note that, by Fermat's little theorem, $\mathsf{Power}_{p-1}(\bm{v})$ is the index vector of $\bm{v}$, i.e.,
\[
    v_i^{p-1} = 
    \begin{cases}
        1, & \text{if $i \in I_{\bm{v}}$}\\
        0, & \text{otherwise.}
    \end{cases}
\]

As in the index compression, we also compress general sparse $N$-vectors by multiplying the same $s\times N$ matrix $\bm{C}$.\footnote{In fact, this need not be the matrix $\bm{C}$. Any $s\times N$ matrix can be used in $\Comp$, as long as any $s$ columns of the matrix are linearly independent. The specific choice of $\bm{C}$, however, allows us to apply the Bj\"orck--Pereyra method and solve the linear system in $O(s^2)$ operations, versus $O(s^3)$ for a generic invertible matrix.}
Hence, the compression of both the general sparse vector and the index vector can be performed simultaneously via homomorphic SIMD (\S\ref{ssec:SIMD}). 
Below is the main theorem for our homomorphic compression scheme.

\begin{algorithm}[t]
\caption{$\Comp$}\label{alg:Comp}
\smallskip
\textbf{Input:} $\ctxt(\bm{d})$\Comment*[r]{ciphertexts of $\bm{d}\in \Z_p^N$}
\smallskip
$\ctxt(\bm{e}) \gets \Eval(\bm{C},\ctxt(\bm{d}))$\Comment*[r]{see Eq.~\eqref{eq:vandermonde} and \S\ref{ssec:notation}}
$\ctxt(\bm{v}) \gets \Eval(\mathsf{Power}_{p-1}, \ctxt(\bm{d}))$\;
$\ctxt(\bm{w}) \gets \mathsf{CompIdx}(\ctxt(\bm{v}))$\Comment*[r]{see Alg.~\ref{alg:CompIdx}}
\smallskip
\textbf{return} $\ctxt(\bm{e}), \ctxt(\bm{w})$\Comment*[r]{ciphertexts of $\bm{e}, \bm{w}\in \Z_p^s$}
\end{algorithm}

\begin{algorithm}[t]
\caption{$\mathsf{Decomp}$}\label{alg:Decomp}
\smallskip
\textbf{Input:} $\ctxt(\bm{e}), \ctxt(\bm{w})$\Comment*[r]{ciphertexts of $\bm{e}, \bm{w}\in \Z_p^s$}
\smallskip
\textbf{Step 1}: Reconstruct $\mathsf{sparse}(\bm{d})$ from $\ctxt(\bm{e})$ and $\ctxt(\bm{w})$.\;
\Indp
$\bm{e} \gets \Dec(\ctxt(\bm{e}))$\;
$I_{\bm{v}} \gets \mathsf{DecompIdx.Step1}(\ctxt(\bm{w}))$\Comment*[r]{see Alg.~\ref{alg:DecompIdx}}
$\mathsf{sparse}(\bm{d}) \gets \mathsf{Reconst}(\bm{e}, I_{\bm{v}})$\Comment*[r]{see Alg.~\ref{alg:Reconst}}
\Indm
\smallskip
\textbf{Step 2}: Reconstruct $\bm{d}$ from $\mathsf{sparse}(\bm{d})$.\Comment*[r]{step for formatting}
\smallskip
\Indp
$\bm{d} \gets (d_1,\cdots, d_N)$,\;
\qquad\qquad where $d_{i} =\begin{cases} x, & (i, x) \in \mathsf{sparse}(\bm{d})\\ 0, & \text{otherwise.} \end{cases}$\;
\Indm
\textbf{return} $\bm{d}\in \Z_p^N$\Comment*[r]{recovered sparse vector}
\end{algorithm}

\begin{theorem}[Homomorphic Compression for Vectors]\label{thm:Comp}
The pair of algorithms $(\Comp,\,\allowbreak \mathsf{Decomp})$ (Algs.~\ref{alg:Comp} and~\ref{alg:Decomp}) is an $N$-to-$2s$ homomorphic compression scheme for $s$-sparse vectors in $\Z_p^N$ (over $\mathcal{E}$). 
Furthermore,
    \begin{enumerate}
        \item $\Comp$ (Alg.~\ref{alg:Comp}) costs $O(N \cdot (s+\log p))$ homomorphic additions and multiplications.
        \item Step 1 of $\mathsf{Decomp}$ (Alg.~\ref{alg:Decomp}) costs $2s$ decryptions and $O(s^2 + s \log^2 s \cdot (\log s + \log p))$ operations in $\Z_p$.
    \end{enumerate}
\end{theorem}
\begin{proof}
The theorem is direct from Lem.~\ref{lem:Decomp} and Thm.~\ref{thm:CompIdx}.
\end{proof}

Our approach uses the compression scheme $(\mathsf{CompIdx}, \mathsf{DecompIdx})$ for sparse index vectors as a building block to compress a sparse vector.
We note that $(\mathsf{CompIdx}, \mathsf{DecompIdx})$ can be replaced by any other sparse index vector compression scheme to construct our sparse vector compression scheme $(\Comp, \mathsf{Decomp})$.

\subsection{Comparison}\label{subsec:comparison}

\begin{table*}
\centering
\caption{Comparison of Homomorphic Compression Schemes
}
\label{tab:comparison}
\Description{A table comparing homomorphic compression schemes for s-sparse vectors of length N by technique, determinism, digest size, compression cost, multiplicative depth, and decompression cost. Our deterministic scheme, based on an extended Newton's identity, has an O(s) digest, O(N(log p + s)) compression cost, and O(log p) depth.}
\begin{tabular}{l l c l l l l}
\toprule
& Technique & Deterministic & Digest Size & $\Comp$ & Multiplicative Depth & $\Decomp$ \\
\midrule
\cite{ToSC:ParCha25} & Sorting & \ding{51} & $O(s)$ & $O(N^2 \log p)$ & $O(N \log p)$ & $O(s)$ \\
\cite{SAC:CGKP24} & Sorting & \ding{51} & $O(s)$ & $O(N \log^2 s \log p)$ & $O(\log N \log s \log p)$ & $O(s)$ \\
\midrule
\cite{CCS:CDGLY21} & Hash & \ding{55} & $O(s(\kappa+\log s))$ & $O(N(\kappa + \log s))$ & $O(1)$ & $O(s(\kappa+\log s))$\\
\cite{C:LiuTro22} & Hash & \ding{55} & $O(\kappa s \log s \log N)$ & $O(N (s+\kappa \log s \log N))$ & $O(1)$ & $O(s^3+\kappa s \log s \log N)$\\
\cite{EC:FleLarSim23} & Hash & \ding{55} & $O(\frac{\kappa s}{\log s})$ & $O(\frac{N\kappa}{\log s})$ & $O(1)$ & $O(\frac{\kappa s}{\log s})$\\
\midrule
\cite{EC:FleLarSim23} & FFT & \ding{51} & $O(s)$ & $O(N \log N)$ & $O(\log N)$ & $O(s\sqrt{N})$\\
\cite{EC:GioGreSim26} & Linear code & \ding{51} & $O(s)$ & $O(N \log^2 s)$\footnotemark[6] & $O(\log s)$ & $O(s \log^2 s \log p)$ \\
Ours & Ext. Newton & \ding{51} & $O(s)$ & $O(N(\log p + s))$ & $O(\log p)$ & $O(s^2+s \log^2 s \log p)$ \\
\bottomrule
\multicolumn{7}{l}{\footnotesize $*$ $N$: length of a sparse vector
\hfill \footnotesize $*$ $s$: sparsity of the sparse vector
\hfill \footnotesize $*$ $\kappa$: failure probability parameter}
\end{tabular}
\end{table*}

Table~\ref{tab:comparison} compares the asymptotic efficiency of homomorphic compression schemes for $s$-sparse vectors in $\Z_p^N$. See also \S\ref{ssec:related}. For a fair comparison, we exclude SIMD optimization and evaluate all schemes over $\Z_p$ arithmetic; in particular, for sorting-based approaches~\cite{SAC:CGKP24, ToSC:ParCha25}, we assume each comparison requires $O(\log p)$ homomorphic operations and a multiplicative depth of $O(\log p)$.

\begin{remark}[Probabilistic vs. Deterministic Schemes]
    In Tab.~\ref{tab:comparison}, probabilistic schemes~\cite{CCS:CDGLY21, C:LiuTro22, EC:FleLarSim23} suffer from less efficient compression rates that grow linearly with $\kappa$, the parameter controlling the failure probability $2^{-\kappa}$. In contrast, deterministic approaches~\cite{EC:FleLarSim23, SAC:CGKP24, ToSC:ParCha25, EC:GioGreSim26}, including ours, achieve asymptotically optimal compression rates. Moreover, since our scheme makes no random choices, the entire compression circuit is deterministic, and its correctness can be independently verified at any time.
\end{remark}

\begin{remark}[HE-Friendliness]\label{rem:SIMD}
A distinctive feature of our scheme is its HE-friendliness, discussed further in \S\ref{ssec:our_PDQ}. First, our scheme is SIMD-friendly (\S\ref{sec:opt_SIMD}). While \cite{CCS:CDGLY21} and \cite{C:LiuTro22} also mention SIMD optimizations, their methods only parallelize small portions of the overall procedure. 
Second, when instantiated in the PDQ setting, our scheme effectively consumes only a single multiplicative depth. 

\end{remark}

\begin{remark}[Plaintext Modulus Constraint]\label{rmk:our_p_geq_N}
The compression algorithm of \cite{EC:GioGreSim26} and ours both require the constraint $p > N$. (See \S\ref{ssec:compress-index}.)
In our scheme, however, this constraint can be relaxed via $p$-ary decomposition of the indices in $[N]$. This incurs only an additional $\log_p N$ factor in both computational complexity and compression rate, which is small in practice.

\end{remark}

\section{Private Database Query}\label{sec:PDQ}

We formalize private database query (PDQ) in \S\ref{ssec:pdq-def}, present a generic construction of PDQ from homomorphic encryption in \S\ref{ssec:generic}, and describe our PDQ scheme derived from our homomorphic compression scheme (\S\ref{sec:homcomp}) in \S\ref{ssec:our_PDQ}.

\subsection{Definition}\label{ssec:pdq-def}

In a \emph{private database query scheme (PDQ)}, a server maintains a database, and users send queries to retrieve records of interest from the server while keeping their queries private.
More precisely, we model a query as a tuple $(P, x, f)$ of a predicate $P$, a condition $x$, and a post-processing circuit $f$. 
The goal of PDQ is to allow the client to retrieve the set of all $f(\DB_i)$ such that $P(x, \DB_i)$ is true while keeping $x$ private from the server.
We formally define PDQ as follows.

\begin{definition}[Private Database Query]\label{def:PDQ}
    Let $\mathcal{D}$ be the set of all possible data points and $\mathcal{F}$ be the set of functions $\mathcal{D} \to \{0,1\}^*$.
    Let $\mathcal{P}$ be the set of predicates $\{0,1\}^* \times \mathcal{D} \rightarrow \{0,1\}$. 
    Define the query space $\mathcal{Q}=\mathcal{P} \times \{0,1\}^* \times \mathcal{F}$.

    A \emph{private database query scheme} is a two-party protocol between a server and a client, where the server holds a length-$N$ database $\mathsf{DB} = (\mathsf{DB}_1, \cdots, \DB_N)\in \mathcal{D}^N$.\footnote{The definition can be easily extended to encrypted $\DB$.}
    A PDQ scheme consists of three algorithms $(\mathsf{Query}, \mathsf{Answer}, \allowbreak \mathsf{Extract})$, described below.
    The usage of these algorithms is described in Fig.~\ref{fig:PDQ_intro}.

    \begin{itemize}
        \item $(\bm{q}, \mathsf{st})\leftarrow\mathsf{Query}(P, x, f)$:
        Takes $(P, x, f)\in \mathcal{Q}$ as input and returns a query $\bm{q}$ together with a client state $\mathsf{st}$.
        \item $\bm{a}\leftarrow\mathsf{Answer}(\mathsf{DB}, \bm{q})$:
        Takes a database $\mathsf{DB}$ and a query $\bm{q}$ as input and returns an answer $\bm{a}$. 
        \item $\bm{r}\leftarrow\mathsf{Extract}(\mathsf{st}, \bm{a})$: 
        Takes an answer $\bm{a}$ together with a client state $\mathsf{st}$ as input and returns $\bm{r}$.
    \end{itemize}
\end{definition}

A PDQ scheme should satisfy the following correctness and privacy guarantees.

\begin{definition}[Correctness]
\label{def:PDQ_correctness}
    A PDQ scheme $(\mathsf{Query}, \mathsf{Answer}, \allowbreak \mathsf{Extract})$ is \emph{correct} if the following probability is overwhelming, for any $(P,x,f)\in\cQ$ and $\DB = (\DB_1, \cdots, \DB_N)\in \cD^N$, where $\bm{r}' = \left\{ 
    f(\DB_i) :
    P(x, \DB_i)=1
    \right\}$.
    \[
    \Pr\left[ 
    \bm{r} = \bm{r}'
    ~\middle|~
    \begin{aligned}
        (\bm{q}, \mathsf{st}) &\leftarrow \mathsf{Query}(P, x, f)\\
        \bm{a} &\leftarrow \mathsf{Answer}(\mathsf{DB}, \bm{q})  \\
        \bm{r} &\leftarrow \mathsf{Extract}(\mathsf{st}, \bm{a})
    \end{aligned}
    \right]
    \]
\end{definition}

\begin{definition}[Privacy]
\label{def:PDQ_privacy}
    A PDQ scheme $(\mathsf{Query}, \mathsf{Answer}, \allowbreak \mathsf{Extract})$ is \emph{private} if $\bm{q}_1 \compind \bm{q}_2$ for any $(P, x_1, f), (P, x_2, f) \in \cQ$ with $|x_1| = |x_2|$, where $(\bm{q}_b, \mathsf{st}_b) \leftarrow \Query(P, x_b, f)$ for $b \in \{1, 2\}$ and $\compind$ denotes computational indistinguishability.
\end{definition}

In some applications, we might also want \emph{server privacy}, i.e., clients should not learn more about the database than they have queried. 
We formalize this requirement as follows.

\begin{definition}[2-Sided Privacy]
\label{def:2-sided}
A correct and private PDQ scheme achieves \emph{$2$-sided privacy} if there exists a PPT simulator $\cS$ satisfying the following,
\[
    \mathtt{view}_{\mathrm{C}}(P,x,f; \DB) \compind \cS(P,x,f,\bm{r}),
\]
where $\bm{r} = \left\{ 
    f(\DB_i) :
    P(x, \DB_i)=1
    \right\}$.
Here, $\mathtt{view}_{\mathrm{C}}(P,x,f; \DB)$ denotes the \emph{view} of the client when the PDQ protocol is run by a client with input $(P,x,f)$ and a server with $\DB=(\DB_1,\cdots,\DB_N)$.
\end{definition}

\begin{remark}\label{rmk:various_query_including_exact}
    Our framework captures various situations and types of queries (e.g., exact-match/range/wildcard queries and conjunction/disjunction of these).
    For example, consider the case of an exact-match query on a key-value database.
    In this case, the database $\DB$ consists of (key, value)-pairs $\DB_i = (k_i, v_i)\in \mathcal{D} = \cK \times \cV$.
    To securely retrieve all the $v_i$s such that $k_i = x \in \cK$, set the predicate $P$ as
    \[
        P(x, (k_i, v_i)) = \begin{cases}
            1 & \text{if } k_i=x\\
            0 & \text{if } k_i\neq x
        \end{cases}
    \]
    and the post-processing circuit $f$ as a projection
    \[
        f(k_i, v_i) = v_i.
    \]
\end{remark}

\footnotetext[6]{The compression complexity of~\cite{EC:GioGreSim26} can be further reduced to $O(N \log s)$ when the plaintext space is a binary extension field. However, binary extension fields do not align well with homomorphic SIMD, making this optimization less attractive in modern HE schemes.}
\setcounter{footnote}{6}

\subsection{Generic Construction from HE}\label{ssec:generic}

We can construct a PDQ scheme $\Pi=(\mathsf{Query}, \allowbreak\mathsf{Answer}, \allowbreak\mathsf{Extract})$ in a generic way from a homomorphic encryption scheme $\mathcal{E}= \allowbreak(\Gen, \allowbreak\Enc, \allowbreak\Dec,\allowbreak \Eval)$ and a homomorphic compression scheme $\cC = (\Comp, \allowbreak\mathsf{Decomp})$ for sparse vectors as follows.

\begin{itemize}
    \item $\mathsf{Query}(P, x, f)$: 
    The client runs $\Gen$ to obtain $(\sk, \pk)$, outputs $\bm{q} \gets (P, \Enc_{\pk}(x), f, \pk)$, and sets $\mathsf{st} \gets \sk$.
    \smallskip
    \item $\mathsf{Answer}(\mathsf{DB}, \bm{q})$: Let $\mathsf{DB}=(\DB_1, \cdots , \DB_N)$. 
    \begin{enumerate}
        \item $\Match$: 
        Identify relevant data points by homomorphic evaluation of $P$,
    \[
        \mathtt{ctxt}(v_i) \leftarrow \Eval_\pk\Big(P(\;\cdot\;, \mathsf{DB}_i),\ctxt(x)\Big).
    \]
        \item $\Mask$: 
        Process data points and zero-ize irrelevant data points.
        This is done by homomorphically multiplying processed data points by the \emph{mask} $\ctxt(v_i)$,
    \begin{align*}
        \ctxt(d_i) &\gets \ctxt(v_i) \boxtimes f(\DB_i).
    \end{align*}
        \item $\Comp$: 
        For efficiency, compress $\ctxt(\bm{d})$ via $\cC=(\Comp, \allowbreak \Decomp)$,
        \[
            \bm{a} \gets \Comp_\pk(\ctxt(\bm{d})). 
        \]
    \end{enumerate}
    \item $\mathsf{Extract}(\mathsf{st}, \bm{a})$: 
    The client extracts the desired result by decompressing $\bm{a}$.
    \[
        \bm{r} \gets \mathsf{Decomp}_\sk(\bm{a}).
    \]
\end{itemize}

\begin{theorem}[Correctness]
    The proposed generic construction $\Pi=(\mathsf{Query}, \allowbreak\mathsf{Answer}, \allowbreak\mathsf{Extract})$ is a correct PDQ scheme if the underlying HE scheme $\cE$ and homomorphic compression scheme $\cC$ are both correct.
\end{theorem}
\begin{proof}
    First, note that $\ctxt(\bm{d})=(\ctxt(d_1), \cdots, \ctxt(d_N))$, the result of $\Mask$ step, encrypts a sparse vector of $f(\DB_i)$ such that $P(x, \DB_i)=1$ by correctness of $\cE$.
    \begin{align*}
        \ctxt(d_i) &= \ctxt(v_i) \boxtimes f(\DB_i)\\
        &= \Eval_\pk\Big(P(\;\cdot\;, \mathsf{DB}_i),\ctxt(x)\Big) \boxtimes f(\DB_i)\\
        &= \ctxt(P(x, \DB_i)) \boxtimes f(\DB_i)\\
        &= \begin{cases}
            \ctxt(f(\DB_i)), &\text{if } P(x, \DB_i)=1\\
            \ctxt(0), &\text{if } P(x, \DB_i)=0.
        \end{cases}
    \end{align*}
    Then, the below follows from the correctness of $\cC$.
    \[
        \Decomp_\sk(\bm{a}) = \mathsf{Decomp}_\sk(\Comp_\pk(\ctxt(\bm{d}))) = \bm{d}
    \]
\end{proof}

\begin{remark}[On Sparsity and Correctness]
To ensure the correctness of our PDQ scheme, it is necessary to set a correct upper bound on the number of relevant items (i.e., sparsity parameter $s$ in \S\ref{sec:homcomp}). Such a bound is often predictable based on the nature of the query and the characteristics of the database. If the number of relevant items exceeds the expected bound, a checksum, for example, can easily detect this. In such cases, the query can be rerun with a larger bound.

\end{remark}

\begin{theorem}[Privacy]
    The proposed generic construction $\Pi=(\mathsf{Query}, \allowbreak \mathsf{Answer}, \mathsf{Extract})$ is a private PDQ scheme if the underlying HE scheme $\cE$ has semantic security. 
\end{theorem}
\begin{proof}
    Directly follows from the semantic security of $\cE$.
\end{proof}

\begin{remark}[2-Sided Privacy]\label{rem:2-sided-privacy}
    With standard circuit-privacy techniques such as \emph{noise-flooding}~\cite{Gen09_thesis} or \emph{sanitization}~\cite{EC:DucSte16} applied to the $\Answer$ step, the proposed generic construction already achieves a relaxed notion of 2-sided privacy (Def.~\ref{def:2-sided}), in which the simulator is given $\{(i, f(\DB_i)) : P(x,\DB_i)=1\}$ instead of $\bm{r}$. This suffices whenever the positions of the matching records carry no sensitive information.

    Full 2-sided privacy (Def.~\ref{def:2-sided}) is then achieved by additionally hiding these positions through shuffling: for each query, the server samples a fresh random permutation $\pi$ of $[N]$ and uses, in place of $\bm{C}$, the column-permuted matrix $\bm{C}_\pi$ whose $i$-th column is $(\pi(i), \pi(i)^2, \cdots,\allowbreak \pi(i)^s)^T$. This shuffling incurs no additional homomorphic computation.
\end{remark}

\subsection{Our PDQ Scheme}\label{ssec:our_PDQ}

At a high level, our PDQ scheme is an instantiation of the above generic construction using our homomorphic compression scheme (\S\ref{sec:homcomp}). While a direct application already inherits asymptotically optimal communication costs due to the optimal compression rate of our scheme, we can further optimize the costs and improve usability under the cost model of modern lattice-based FHE schemes such as BFV~\cite{C:Brakerski12,EPRINT:FanVer12}. These refinements yield substantially better practical performance and make our approach more attractive in real-world settings.

\begin{algorithm}[t]
\caption{$\mathsf{Answer}$}\label{alg:New_Answer}
\smallskip
\textbf{Input:} $\mathsf{DB}=(\DB_1, \cdots, \DB_N)$, $P:\{0,1\}^* \times \mathcal{D} \rightarrow \{0,1\}$, $\ctxt(x)$, $f:\mathcal{D} \to \{0,1\}^*$\Comment*[r]{database, predicate, encrypted query, and post-processing function, resp.}
\textbf{Step 1}: $\mathsf{Match}$\;
\Indp
\For{$i\in [N]$}{
$\ctxt(v_i) \gets \Eval_\pk\Big(P(\;\cdot\;, \mathsf{DB}_i),\ctxt(x)\Big)$\;
}
$\ctxt(\bm{v}) \gets (\ctxt(v_1), \cdots , \ctxt(v_N))$\;
\Indm
\textbf{Step 2}: $\Mask$\;
\Indp
\For{$i\in [N]$}{
$\ctxt(d_i) \gets \ctxt(v_i) \boxtimes f(\DB_i)$\;
}
$\ctxt(\bm{d}) \leftarrow (\ctxt(d_1), \cdots, \ctxt(d_N))$\;
\Indm
\textbf{Step 3}: $\Comp$\;
\Indp
$\ctxt(\bm{e})\gets \Eval_\pk (\bm{C}, \ctxt(\bm{d}))$\;
$\ctxt(\bm{w})\gets \mathsf{CompIdx}(\ctxt(\bm{v}))$\Comment*[r]{see Alg.~\ref{alg:CompIdx}}
\Indm
\smallskip
\textbf{return} $\ctxt(\bm{e}), \ctxt(\bm{w})$\Comment*[r]{compressed PDQ answer}
\end{algorithm}

\subsubsection{Reusing Index Vector}\label{sssec:opt_reuse}

Recall that our compression algorithm involves homomorphically computing the index vector from the sparse vector. 
However, in the context of PDQ, we already have the corresponding index vector for the sparse vector (or masked database) from $\mathsf{Match}$. 
The punchline is that, in natural PDQ protocols, the sparse vector is actually obtained \emph{from} the index vector through masking ($\mathsf{Mask}$). 
Thus, we can skip the homomorphic computation of the index vector and reuse the result of the $\mathsf{Match}$ step.
We believe this trick is also applicable to other scenarios where homomorphic compression is required and that our compression scheme will be generally useful.
The modified $\Answer$ algorithm is described in Alg.~\ref{alg:New_Answer}.
Observe that our homomorphic compression scheme now reduces to two homomorphic matrix-vector multiplications (namely, $\bm{C}\bm{v}$ and $\bm{C}\bm{d}$) and has an \emph{effective} multiplicative depth of one. This is significant: the depth of the compression circuit directly drives parameter selection for the entire PDQ pipeline, so lower depth yields lighter parameters and better overall performance.

\subsubsection{SIMD-Aware Implementation}\label{sec:opt_SIMD}
A distinguishing feature of our PDQ scheme is its ability to fully leverage the homomorphic SIMD capabilities of modern FHE schemes, particularly BFV (\S\ref{ssec:SIMD}). Since the $\mathsf{Match}$ and $\mathsf{Mask}$ steps readily exploit the SIMD structure, we focus on the SIMD-aware implementation of our homomorphic compression scheme (i.e., the $\Comp$ step in PDQ).

Recall that, after applying the index-vector reuse optimization (\S\ref{sssec:opt_reuse}), our homomorphic compression algorithm (Alg.~\ref{alg:New_Answer}) essentially reduces to two homomorphic matrix-vector multiplications, namely $\bm{C}\bm{v}$ and $\bm{C}\bm{d}$.

Among several works \cite{C:HalSho14, USENIX:JuvVaiCha18, SP:LHHMQ21-fixed} on homomorphic matrix-vector multiplications, we follow the \textsc{Pegasus} \cite{SP:LHHMQ21-fixed} framework, which targets \emph{short} matrices as in our case ($\bm{C}\in\Z_p^{s\times N}$ with $s \ll N$).
More specifically, the framework requires $\tilde{O}(\sqrt{sN/n})$ key-switching operations, which is sublinear in the sparsity $s$. 
While the original \textsc{Pegasus} is designed for the CKKS scheme~\cite{AC:CKKS17}, which has a plaintext space $\C^{n/2}$, we adapt \textsc{Pegasus} for simultaneous evaluation of two matrix-vector multiplications over the BFV scheme with plaintext space $\Z_p^{2\times(n/2)}$.

\subsubsection{Ring-Switching}\label{subsec:RingSwitching}

To support realistic queries in the $\mathsf{Match}$ step, we must employ large HE parameters (e.g., $n = 2^{15}$) that enable \emph{deep} homomorphic computations.
On the other hand, our homomorphic compression scheme compresses the masked database into a single ciphertext whenever $2s \le n$. Thus, when $2s < n$, message slots are underutilized while the communication cost remains unchanged. Moreover, large HE parameters significantly slow down the overall PDQ protocol, particularly the $\Comp$ step, since the performance of lattice-based HE schemes typically degrades superlinearly in $n$.

To address this issue, we adopt the \emph{ring-switching} technique~\cite{ITCS:BraGenVai12, SCN:GHPS12, GHPS13} to transition to smaller HE parameters after the $\mathsf{Match}$ and $\mathsf{Mask}$ steps. Specifically, ring switching reduces the ring degree $n$ of a given HE ciphertext to a smaller value. This can be achieved while preserving the same security level by applying modulus switching to reduce the ciphertext modulus accordingly. Since our homomorphic compression scheme effectively requires only a single multiplicative depth (\S\ref{sssec:opt_reuse}), we can safely switch to a much smaller ciphertext modulus and, consequently, a smaller ring degree (e.g., $n = 2^{13}$).

\begin{remark}
    Prior work, especially~\cite{CCS:CDGLY21}, does not explicitly address these issues. In particular, it is unclear how large-depth computations in the $\mathsf{Match}$ step can be supported while keeping HE parameters small (e.g., $n = 2^{13}$).
\end{remark}

\section{Experimental Results}\label{sec:experiment}

To evaluate the concrete performance of our SIMD-aware homomorphic compression scheme (\S\ref{sec:homcomp}), we implement and benchmark a PDQ protocol as a proof of concept (\S\ref{ssec:experiment_setup} and~\S\ref{ssec:benchmark}), and compare it with prior work (\S\ref{ssec:comparison}). The code is available at {\color{blue} \url{https://github.com/keewoolee/PDQ}}. Our implementation supports exact-match queries on a key-value database, where each key and value fit into a single message slot of the underlying BFV scheme. 

Note that the end-to-end performance of PDQ heavily depends on the type of query performed in the $\mathsf{Match}$ step. Since homomorphic compression is our primary focus, we choose a toy query as a representative example. Nevertheless, this demonstrates that compression constitutes a non-negligible portion of the end-to-end runtime, particularly in prior work with less efficient compression schemes, and highlights how improvements in compression can significantly impact overall performance. (See Rem.~\ref{rem:ratio}.) For a more interesting application of our compression scheme to PDQ that involves a more complex $\mathsf{Match}$ step, we refer the reader to \S\ref{sssec:OMR}.

\subsection{Experimental Setup}\label{ssec:experiment_setup}

\subsubsection{Methodology}\label{sssec:Environments}

We implement our PDQ protocol in C++ using the OpenFHE library~\cite{WAHC:BBBCEG22} for the BFV scheme and the NTL library~\cite{NTL} to implement $\mathsf{Decomp}$. We leverage the HEXL library~\cite{HEXL} to accelerate HE operations. All experiments were conducted on Ubuntu 24.04 LTS using a Google Compute Engine \texttt{n2-standard-16} instance (Intel Ice Lake @ 2.6 GHz, 64 GB RAM), with execution restricted to a single thread by setting \texttt{OMP\_NUM\_THREADS=1}. Each experiment was repeated five times, and we report the median.

\paragraph{Match Step} 
For the $\Match$ step, we implement the exact-match predicate as follows. Let $\DB = \{(k_i, v_i)\}_{i=1}^N$ be a key-value database in which each $k_i$ and $v_i$ fits into a single ciphertext slot, and let $\Enc(k)$ be an encryption of the query key $k$. We homomorphically evaluate the expression $1 - (k - k_i)^{p-1}$, where $p$ denotes the plaintext modulus. By Fermat’s little theorem, this evaluates to $1$ when $k = k_i$ and to $0$ otherwise.

\paragraph{Encrypted DB}
We consider the setting in which the database values $k_i$ and $v_i$ are stored in BFV-encrypted form. Compared to storing the database in plaintext, this slightly degrades the performance of $\Match$, $\Mask$, and the ring-switching step. In contrast, the $\Comp$ and $\Decomp$ steps are unaffected.

\subsubsection{Parameters}\label{sssec:param}

\begin{table}
\caption{BFV Parameters. $Q$ and $P$ denote the ciphertext modulus and the auxiliary modulus for key switching.}\label{tab:params}
\Description{A table of the BFV parameter sets, listing the ring dimension n, the plaintext modulus p, log Q, and log PQ. The main sets param-1 (n = 32768, p = 65537) and param-2 (n = 65536, p = 786433) each come with a ring-switching set of dimension 8192 and 16384, respectively.}
\centering
\begin{tabular}{l r r r r}
\toprule
&
\multicolumn{1}{c}{$n$} &
\multicolumn{1}{c}{$p$} &
\multicolumn{1}{c}{$\log Q$} &
\multicolumn{1}{c}{$\log PQ$} \\
\midrule
\texttt{param-1}   &  32768 & 65537 & 660 & 840 \\
\texttt{param-1-rs} & 8192 & 65537 & 120 & 180 \\
\midrule
\texttt{param-2}  &  65536 & 786433 & 900 & 1380 \\
\texttt{param-2-rs} & 16384 & 786433 & 180 & 360 \\
\bottomrule
\end{tabular}
\end{table}

Our implementation is based on the BFV scheme with the parameter sets described in Tab.~\ref{tab:params}. We use \texttt{param-1} for $\Match$ and \texttt{param-1-rs} for ring switching and subsequent steps when $N < 65537$, and \texttt{param-2} and \texttt{param-2-rs}, respectively, when $N \geq 65537$. This accounts for the condition $p > N$ when using our SIMD-aware homomorphic compression scheme (\S\ref{sec:homcomp}). All parameter sets, including the ring-switching parameters, achieve security levels of at least 128 bits, as per \cite{HES}.

\subsection{Benchmark Results}\label{ssec:benchmark}

In this section, we present benchmark results for PDQ. All experiments follow the experimental setup described in \S\ref{ssec:experiment_setup}. Results for varying database sizes $N$ are presented in \S\ref{sssec:varying_N}, results for varying sparsity $s$ are discussed in \S\ref{sssec:varying_s}, and communication costs are analyzed in \S\ref{sssec:comm_cost}.

\subsubsection{Runtime with Varying $N$}\label{sssec:varying_N}

\begin{table}
\centering
\caption{Runtime with Varying $N$ ($s=16$). The column RS denotes the ring-switching runtime.}
\label{tab:varying_N}
\Description{A table of the runtimes of Match, Mask, ring switching, and Comp in seconds and of Decomp in milliseconds, for s = 16 and database sizes N from 8192 to 524288. Match dominates and rises from 2.19 to 60.77 seconds as N grows, while Decomp stays between 1.25 and 1.76 milliseconds.}
\begin{tabular}{r S S S S S}
\toprule
\multicolumn{1}{c}{$N$} &
\multicolumn{1}{c}{\makecell{$\Match$\\(sec)}} &
\multicolumn{1}{c}{\makecell{$\Mask$\\(sec)}} &
\multicolumn{1}{c}{\makecell{$\mathsf{RS}$\\(sec)}} &
\multicolumn{1}{c}{\makecell{$\Comp$\\(sec)}} &
\multicolumn{1}{c}{\makecell{$\Decomp$\\ (ms)}} \\
\midrule
8192      &   2.19     &      0.06     &      0.14     &      0.12     &      1.25\\
16384     &   2.18     &      0.06     &      0.14     &      0.12     &      1.28\\
32768     &   2.16     &      0.06     &      0.15     &      0.12     &      1.29\\
65536     &   4.28     &      0.11     &      0.28     &      0.17     &      1.28\\
\midrule
131072    &   15.28    &       0.35    &       0.87    &       0.53    &       1.71\\
262144    &   30.71    &       0.71    &       1.78    &       0.65    &       1.64\\
524288    &   60.77    &       1.44    &       3.51    &       1.05    &       1.76\\
\bottomrule
\end{tabular}
\end{table}

Table~\ref{tab:varying_N} reports the results for a fixed sparsity ($s = 16$) while varying the database size ($2^{13} \le N \le 2^{19}$). Note that we use two different sets of BFV parameters (\S\ref{sssec:param}): \texttt{param-1} for $2^{13} \le N \le 2^{16}$ and \texttt{param-2} for $2^{17} \le N \le 2^{19}$.
From the table, we observe that the server runtime grows linearly with the number of ciphertexts required to pack the entire database, i.e., $\lceil N / n \rceil$, when the same parameters are used. Note that, for $\Comp$, there is an additional constant overhead associated with formatting the packing structure, which scales with $\log(n/s)$.
On the other hand, for a fixed sparsity $s$, the decompression latency remains constant, up to minor noise.

\subsubsection{Runtime with Varying $s$}\label{sssec:varying_s}

\begin{figure}
\centering
\includegraphics{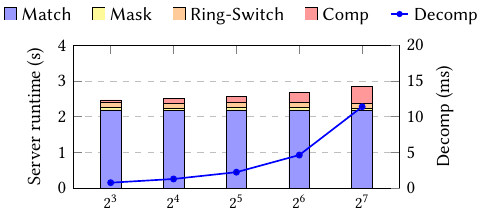}
\caption{Runtime with Varying $s$ ($N=16384$)}
\Description{A stacked bar chart of the server runtime for N = 16384 and s from 8 to 128, in which Match dominates at about 2.2 seconds and the total grows only from about 2.5 to 2.8 seconds. A line on a second axis shows the Decomp time rising from about 0.8 to 11.4 milliseconds.}
\label{fig:varying_s_wide}
\end{figure}

Figure~\ref{fig:varying_s_wide} reports the results for a fixed database size ($N=16384$) under varying sparsity ($2^3 \le s \le 2^7$). We observe that the server runtime, excluding $\Comp$, is independent of the sparsity $s$, as it depends solely on the database size $N$.

The figure shows that $\Comp$ scales sublinearly. This behavior reflects the fact that $\Comp$ essentially performs two homomorphic matrix multiplications of dimension $s \times N$, for which the number of key-switching operations scales sublinearly in the sparsity $s$ in the \textsc{Pegasus} framework (\S\ref{sec:opt_SIMD}).

On the other hand, we observe that the runtime of $\Decomp$ grows superlinearly with the sparsity $s$, reflecting the quadratic decompression complexity stated in Thm.~\ref{thm:Comp}. However, this quadratic overhead arises from plaintext operations performed after decryption rather than from homomorphic computation. Consequently, the constant factors are small, and the $\Decomp$ runtime remains below 15 ms even for $s = 128$.

\subsubsection{Communication Cost}\label{sssec:comm_cost}

\begin{table}
\centering
\caption{Communication Cost. The columns evk$_\mathrm{R}$ and evk$_\mathrm{S}$ denote the relinearization key required for homomorphic multiplications and the key-switching key required for ring switching, respectively.}
\label{tab:PDQ_comm_cost}
\Description{A table of the digest size, the query size, and the sizes of the relinearization and ring-switching keys for the two parameter sets. These are 129 KB, 5.5 MB, 28.8 MB, and 3.0 MB with param-1, and 257 KB, 15.0 MB, 47.5 MB, and 6.0 MB with param-2.}
\begin{tabular}{c c c c c}
\toprule
& {Digest~(KB)}  &   {Query~(MB)} &  {evk$_\mathrm{R}$~(MB)}  
&   {evk$_\mathrm{S}$~(MB)}\\
\midrule
\texttt{param-1}       &       129      &      ~~5.5      &     28.8      
&     3.0\\
\texttt{param-2}       &       257      &     15.0      &     47.5      
&     6.0\\
\bottomrule
\end{tabular}
\end{table}

\begin{table}
\centering
\caption{Rotation Key Size with Varying $s$ ($N=2^{14}$)}
\label{tab:PDQ_comm_rotkey_cost_vary_s_transposed}
\Description{A table of the rotation key size with param-1 for N = 16384 and s from 8 to 128. The size decreases from 8.3 MB to 6.0 MB as s grows.}
\begin{tabular}{c c c}
\toprule
& $s$ & evk$_{\mathrm{rot}}$ (MB) \\
\midrule
\multirow{5}{*}{\texttt{param-1}} & $2^3$ & 8.3 \\
& $2^4$ & 8.3 \\
& $2^5$ & 7.5 \\
& $2^6$ & 6.8 \\
& $2^7$ & 6.0 \\
\bottomrule
\end{tabular}
\end{table}

\begin{table}
\centering
\caption{Rotation Key Size with Varying $N$ ($s=16$)}
\label{tab:PDQ_comm_rotkey_cost_vary_N_transposed}
\Description{A table of the rotation key size for s = 16 and database sizes N from 8192 to 524288. The size is between 7.5 and 8.3 MB with param-1 and 16.5 MB with param-2.}
\begin{tabular}{c c c}
\toprule
& $N$ & evk$_{\mathrm{rot}}$ (MB) \\
\midrule
\multirow{4}{*}{\texttt{param-1}} & $2^{13}$ & ~8.3 \\
& $2^{14}$ & ~8.3 \\
& $2^{15}$ & ~8.3 \\
& $2^{16}$ & ~7.5 \\
\midrule
\multirow{3}{*}{\texttt{param-2}} & $2^{17}$ & 16.5 \\
& $2^{18}$ & 16.5 \\
& $2^{19}$ & 16.5 \\
\bottomrule
\end{tabular}
\end{table}

Table~\ref{tab:PDQ_comm_cost} reports the digest size, query size, and key sizes, which depend only on the chosen HE parameter set and are independent of other configurations ($N$ and $s$). These include the relinearization key required for homomorphic multiplications and the key-switching key required for ring switching. The rotation key size, which varies with $N$ and $s$, is reported separately in Tabs.~\ref{tab:PDQ_comm_rotkey_cost_vary_s_transposed} and~\ref{tab:PDQ_comm_rotkey_cost_vary_N_transposed}. Key transmission is a one-time setup cost that can be amortized across multiple queries.

\begin{remark}[Digest Size Optimization]\label{rem:opt_digest}
    We can further reduce the digest size by setting the BFV ciphertext modulus at the lowest level to be as small as possible. However, we do not apply this optimization in our implementation due to limitations of the OpenFHE library: (i) it does not officially support setting the base-level ciphertext modulus independently of the RNS limbs, and (ii) its ciphertext serialization does not shrink further when the ciphertext modulus is below 64 bits. If this optimization were applied, the digest size could be reduced by approximately an additional factor of 2.
\end{remark}

\begin{remark}[Key Size Optimization]
    We note that the key sizes (as well as the query size) can be reduced by approximately half by applying hash-based compression, which compresses fresh randomness into a hash seed.
\end{remark}

\begin{remark}[Query Size Optimization]
    We note that the focus of our proof of concept is on homomorphic compression rather than query optimization; accordingly, we have not optimized the query size. The query size could be significantly reduced using techniques such as transciphering~\cite{CCSW:NaeLauVai11_transciphering, EC:ARSTZ15, C:DEGLLL18}.
\end{remark}

\subsection{Comparison to Prior Work}\label{ssec:comparison}

\subsubsection{Comparison with \cite{CCS:CDGLY21}}\label{sssec:comp_CDGLY}

\begin{figure*}
\centering

\begin{minipage}[]{0.48\linewidth}
\centering
\includegraphics{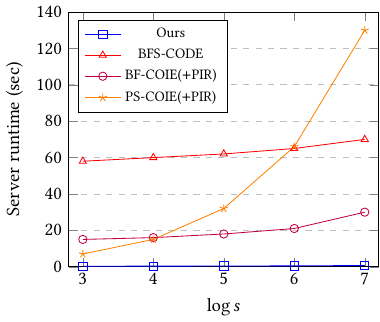}

{\small\textbf{(a)} Varying sparsity $s$ ($N=16384$).}
\end{minipage}
\begin{minipage}[]{0.48\linewidth}
\centering
\includegraphics{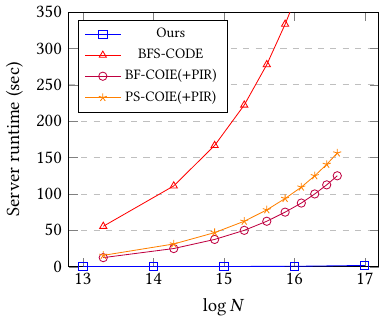}

{\small\textbf{(b)} Varying database size $N$ ($s=16$).}
\end{minipage}
\caption{Comparison with \cite{CCS:CDGLY21} under varying sparsity $s$ and database size $N$
}
\Description{Two line charts comparing the server runtime of our scheme with BFS-CODE, BF-COIE with PIR, and PS-COIE with PIR, for varying s (left) and varying database size (right). Our runtime stays below 2 seconds, while the prior schemes take from 7 seconds to over 350 seconds.}
\label{fig:graph_varying_s_N}
\end{figure*}

An experimental comparison between our scheme and \cite{CCS:CDGLY21} is shown in Fig.~\ref{fig:graph_varying_s_N}. All performance numbers for \cite{CCS:CDGLY21} are taken directly from the paper, as their implementation is not publicly available. In \cite{CCS:CDGLY21}, all experiments were conducted on an Intel Core i9-9900K @ 4.7 GHz with 64 GB of memory, which exceeds our hardware specifications, and the networking protocol between the server and clients uses a 1 Gbps LAN. 
The experiments are compared under the same ring degree $n = 2^{13}$.\footnote{Although the HE parameters are not explicitly specified in \cite{CCS:CDGLY21}, the reported ciphertext size of 103 KB suggests that $n = 2^{13}$.} 
Even under experimental conditions that favor \cite{CCS:CDGLY21}, our scheme outperforms their homomorphic compression scheme BFS-CODE, achieving a $106\times$–$595\times$ speedup. Moreover, their digest size is 206 KB, whereas ours is 129 KB and can be further reduced (Rem.~\ref{rem:opt_digest}).

Choi et al.~\cite{CCS:CDGLY21} also proposed compression schemes for \emph{index vectors}, namely BF-COIE and PS-COIE, which can be combined with PIR schemes to retrieve data after decompressing the matching indices (see \S\ref{sssec:related_comp}). Compared with BF-COIE and PS-COIE combined with PIR, our scheme achieves a $25\times$–$197\times$ speedup.
We note that this approach incurs additional rounds and higher communication costs due to the use of an extra PIR protocol. Moreover, it introduces the possibility of a malicious client retrieving data other than that corresponding to the decompressed indices.

\begin{remark}[Compression Cost Relative to Server Runtime]\label{rem:ratio}     
According to our benchmark results (\S\ref{ssec:benchmark}), the compression step accounts for 9.0\%–23.2\% of the total PDQ runtime, while $\Match$ consistently dominates. Although this may suggest that compression is a minor component of the protocol, a rough estimate indicates that replacing our scheme with the compression method of \cite{CCS:CDGLY21}, which is at least $106\times$ slower, would result in compression consuming 91\%–97\% of the total runtime.
While this estimate is based on our toy example, it nevertheless demonstrates that compression can become a bottleneck when instantiated with less efficient schemes, as illustrated by the Oblivious Message Retrieval (OMR) example in \S\ref{sssec:OMR}.    
\end{remark}

\subsubsection{Comparison with SHELDB~\cite{ToSC:ParCha25}}\label{sssec:comp_SHELDB}

Parbat and Chatterjee~\cite{ToSC:ParCha25} construct PDQ using a sorting-based homomorphic compression scheme over TFHE~\cite{JC:CGGI20}. However, since TFHE does not support homomorphic SIMD optimization, it incurs substantial computational and communication overhead.
In the 16-bit key-value exact-match setting, using an Intel Core i7-7700 @ 3.60 GHz with 64 GB of memory, their $\Comp$ step alone takes 46 minutes, while their $\Match$ step requires more than 16 days for $(N, s) = (768, 21)$. In contrast, our PDQ scheme processes the same 16-bit key-value exact-match setting with $(16384, 32)$ in 0.37 seconds for $\Comp$ and 2.19 seconds for $\Match$, yielding a speedup of more than $500{,}000\times$, modulo differences in CPU architecture.
Furthermore, while their digest size for $(N,s)=(768, 21)$ is 967 KB, ours is only 129 KB for $(16384, 32)$, a reduction of approximately $7.5\times$.\footnote{We note that~\cite{ToSC:ParCha25} does not report the evaluation key sizes required for transmission.}

Even with the techniques proposed in~\cite{SAC:CGKP24}, which study homomorphic \emph{selection} rather than sorting in the context of $k$-nearest neighbors, the computational overhead remains substantial. The reported $k$-NN implementation requires 642.3 seconds for $(N, s) = (1000, 5)$ and 441.1 seconds for $(1000, 3)$ on an Intel Core i9-9900 @ 3.10 GHz. Conservatively, at least 200 seconds are spent on the selection (compression) phase alone for $(N, s) = (1000, 5)$. By comparison, our scheme processes $(8192, 16)$ in 0.32 seconds, a speedup of at least $600\times$, modulo differences in CPU architecture.

\subsubsection{Comparison with \cite{C:LiuTro22}}\label{sssec:OMR}

Although we do not provide a direct performance comparison with~\cite{C:LiuTro22} due to technical constraints, we refer to the follow-up work~\cite{USENIX:LeeYeo26}, which builds upon our compression scheme and demonstrates substantial performance improvements for the Oblivious Message Retrieval (OMR) scheme introduced in~\cite{C:LiuTro22}.

OMR~\cite{C:LiuTro22} is a cryptographic technique that allows a receiver to retrieve his pertinent messages from a bulletin board without revealing which messages are intended for him. OMR can be viewed as a special case of PDQ, where the predicate in the $\Match$ step homomorphically checks the \emph{private signal} (key) attached to each message (value) on the bulletin board against the receiver’s public key to determine whether the message is relevant to that receiver. A follow-up work of ours, SophOMR~\cite{USENIX:LeeYeo26}, leveraged our SIMD-aware compression technique and achieved a $7.5\times$ reduction in compression latency, which leads to a $3.4\times$ overall improvement in OMR latency compared to the prior state-of-the-art OMR scheme, PerfOMR~\cite{USENIX:LiuTroWan24}.

\begin{acks}
This research was supported by Policing Data Utilization Technology Development Program through the Korea Institutes of Police Technology (KIPoT) funded by the Korean National Police Agency (RS-2026-25529965).
\end{acks}

\bibliographystyle{ACM-Reference-Format}
\bibliography{bib/abbrev0, bib/crypto, bib/extra}

\appendix

\section*{Open Science}
In line with open science principles, we are committed to making our research fully transparent and accessible. To support reproducibility and independent verification, we include comprehensive details of our experimental setup and results in the paper, along with the accompanying code available at: {\color{blue} \url{https://github.com/keewoolee/PDQ}}.

\section*{Generative AI Usage}

Generative AI was used for editorial purposes in this manuscript, and all outputs were inspected by the authors to ensure accuracy and originality. We also used an AI coding assistant to help write the implementation and benchmark scripts. The authors reviewed and tested all of this code, and all reported results come from running it. Generative AI is not an object of study in this work, and no models were trained.

\section*{Ethical Considerations}
We analyze the ethical implications of our work following principles from the Menlo Report. Our work does not involve real user data and is evaluated on non-sensitive datasets. We clearly state assumptions and limitations of our methods to prevent misuse or misinterpretation. We conclude that the benefits of enabling practical privacy protections outweigh the risks, especially when combined with appropriate system-level controls.

Stakeholders include end users issuing private queries, data owners and service providers, organizations deploying privacy-preserving systems, and society at large.

\paragraph{Benefits} 
Our PDQ construction improves the efficiency of homomorphic encryption-based query processing, making strong query privacy more practical. This supports user autonomy and helps prevent profiling or leakage of sensitive interests, aligning with the principles of beneficence and respect for persons.

\paragraph{Potential Harms}
The main risk is dual use: privacy-preserving queries may also be used to conceal malicious intent, potentially reducing auditability or enabling abuse (e.g., probing sensitive databases). There are also risks from incorrect implementation or overestimating security guarantees.

\end{document}